\documentclass[runningheads]{llncs}
\usepackage[latin1]{inputenc}   
\usepackage[T1]{fontenc}  
\usepackage{amsmath}
\usepackage{prooftree}
\usepackage{hyperref}
\usepackage{proof}
\usepackage{qtree}
\usepackage{enumitem}

\usepackage{color}
\usepackage{graphicx}
\usepackage{amsfonts}
\usepackage{todonotes}

\newcommand {\hide}[1]{}

\newcommand{\irule}[3]
{\prooftree{#1}\justifies{#2}\using{\:#3}\endprooftree}
\DeclareSymbolFont{symbolsC}{U}{txsyc}{m}{n}
\DeclareMathSymbol{\strictif}{\mathrel}{symbolsC}{74}
\DeclareMathSymbol{\boxright}{\mathrel}{symbolsC}{128}

\newcommand{\Bet}{{\mathcal Bet}}

\newcommand{\hh}{\, | \, }
\newcommand{\seq}{\Rightarrow}
\newcommand{\lam}{\lambda}
\newcommand{\De}{\Delta}
\newcommand{\Ga}{\Gamma}
\newcommand{\Si}{\Sigma}
\newcommand*{\mg}{G}
\newcommand*{\mh}{H}
\newcommand*{\h}{\mid}

\newcommand {\ri} {\rightarrow}
\newcommand {\Ri} {\Rightarrow}

\newcommand {\bes} {\begin{description}}
	\newcommand{\ens} {\end{description}}

\newcommand {\beq} {\begin{quote}}
	\newcommand {\enq} {\end{quote}}
\newcommand {\bit} {\begin{itemize}}
	\newcommand {\enit} {\end{itemize}}
\renewcommand {\neg} {\lnot}
\newcommand {\caL} {{\cal L}}

\newcommand{\be}{\begin{enumerate}}
	\newcommand{\ee}{\end{enumerate}}

\newcommand{\hgt}[1]{\ensuremath{|#1|}}
\newcommand{\cp}[1]{\ensuremath{\ulcorner #1
\urcorner}}

\newcommand{\E}{{\bf E}}
\newcommand{\5}{S5}
\newcommand{\HE}{{\bf HE}}

\usepackage{amssymb,latexsym}
\usepackage{stmaryrd}
\usepackage{gn-logic14}

\title{Dyadic Obligations: Proofs and Countermodels via Hypersequents}
\author{Agata Ciabattoni\inst{1}\orcidID{0000-0001-6947-8772} \and
Nicola Olivetti\inst{2}\orcidID{0000-0001-6254-3754} \and
Xavier Parent\inst{1}\orcidID{0000-0002-6623-9853}}

\authorrunning{A. Ciabattoni et al.}
\institute{TU Wien, Vienna, Austria\\
\email{\{agata, xavier\}@logic.at}
\and
Aix-Marseille Univ, Universit\'{e} de Toulon, CNRS,
LIS, Marseille, France, \\
\email{nicola.olivetti@univ-amu.fr}}

\begin{document}

\maketitle
\begin{abstract}
The basic system $\E$ of dyadic deontic logic proposed by \AA qvist offers a simple solution to contrary-to-duty paradoxes and allows to represent norms with exceptions. We 
investigate $\E$ from a proof-theoretical viewpoint. We propose a hypersequent calculus with 
good 
properties, the most important of which is cut-elimination, and the consequent subformula property. The calculus is refined to obtain a decision procedure for $\E$ and an
effective countermodel computation in case of failure of proof search. 
By means of the refined calculus, we prove that validity in $\E$ is Co-NP and countermodels have  polynomial size.
\end{abstract}
\section{Introduction}

Deontic logic deals with obligation and other normative concepts, which are 
important in 
a variety of fields---from law and ethics to artificial intelligence.

%
Obligations are contextual in nature, and take the form of conditional statements ("if-then"). Their formal analysis rely on dyadic deontic systems. The family of those systems that come with a "preference-based" semantics is the best known one. It was originally developed by \cite{ddl:danielsson1968preference,ddl:H69},
and adapted to a modal logic setting by \AA qvist \cite{A84} and
Lewis~\cite{ddl:L73}. The framework has roots in the so-called
classical theory of rational choice, sharing the
assumption that 
a normative judgment is based on a maximization process of normative preferences.  In that
framework, $\bigcirc (B/A)$ (reading: "$B$ is obligatory, given $A$") is true when the best $A$-worlds are all
$B$-worlds.  The framework was early recognized as a landmark, due to
its ability to handle at once two different kinds of deontic
conditionals, whose treatment within a usual Kripke semantics had
proved elusive: (a) Contrary-to-duty (CTD) conditionals, and (b)
Defeasible deontic conditionals. 
The former are obligations that come into force when some other obligation
is violated. 
As is well-known (e.g.~\cite{Forr84}), deontic logicians have struggled with the the problem of giving a formal
treatment to CTD obligations. According to Hansson~\cite{ddl:H69},
	van Fraassen~\cite{ddl:vanFraassen1972}, Lewis~\cite{ddl:L73} and others,
	the problems raised by CTDs call for an ordering on
	possible worlds in terms of preference (or relative goodness, or betterness),
	and Kripke-style models fail in as much as they do not allow
	for grades of ideality.  
	The use of a
	preference relation has also been advocated for the
	analysis of defeasible conditional obligations.
	In particular, Alchourr\'{o}n~\cite{ddl:Al94} argues that
	preferential models provide a better treatment of this notion
	than the usual Kripke-style models.  Indeed, a defeasible
	conditional obligation leaves room for exceptions.
	Under a preference-based approach, we no longer have the
	deontic analogue of two laws, the failure of which constitutes
	the main formal feature expected from defeasible
	conditionals; these are "deontic" modus-ponens (or
	Factual Detachment): $\bigcirc(B/A)$ and $A$ imply
	$\bigcirc B$, and Strengthening of the
	Antecedent: $\bigcirc(B/A)$ entails $\bigcirc(B/A\wedge C)$.
(There is an extensive literature on the treatment of contrary-to-duties, e.g.~\cite{ddl:vanFraassen1972,ddl:L73,ddl:T81,ddl:LB83,ddl:PS97}, and defeasible conditional obligations, e.g.~\cite{ddl:M93,ddl:AB97,ddl:TT97,ddl:horty14},
in a preference-based framework.) 
%

The meta-theory of the framework has been the focus of much research in recent years (for an overview, see~\cite{ddl:P21}). Like in traditional modal logic, different properties
of the relation in the models yield different Hilbert systems. Early axiomatisation 
results~\cite{ddl:vanFraassen1972,ddl:S75,ddl:L73} 
were tailored to the case where the betterness relation comes with many properties.
These have been criticized as being too demanding in some contexts. Therefore subsequent research investigated how to extend these results to models equipped with a betterness relation meeting
less conditions, if any at
all~\cite{ddl:lou19,Parent15}.  
\AA qvist's system {\bf E}, corresponds to the most general case, involving no commitment to any structural property of the relation. Stronger systems--like
{\bf F} and {\bf G}--are obtained by adding extra constraints on the betterness relation. (A roadmap of existing systems is, e.g., in~\cite{ddl:lou19,ddl:P21}.)
In this paper we focus on {\bf E}, the weakest known
preference-based dyadic deontic logic.


So far for preference-based deontic logics there has been an almost exclusive focus on the  connection between semantic properties and Hilbert systems. Very little research has been done on Gentzen-style calculi. To our knowledge only {\bf G}, due to its equivalence with Lewis's VTA and van Fraassen's CD, has an analytic Gentzen calculus~\cite{GLOP2016}. As is well known such calculi have significant practical and theoretical advantages compared to Hilbert systems.  
In analytic calculi proof search proceeds indeed by
step-wise decomposition of the formulas to be proven.
For this reason 
they can be employed to establish important meta-logical properties
for the formalized logics (e.g., decidability, complexity and interpolation), and facilitate the
development of automated reasoning methods.  In general, analytic
calculi serve to find derivations and hence provide forms of
constructive {\em explanations} for normative systems; e.g.~showing which hypotheses have been used in deriving certain obligations given specific facts. 
They also facilitate counter-model construction from
non-derivable statements, and hence provide explanations of why "something should not be done".

The present paper aims at filling in this gap, 
focusing on \AA qvist's system \E. We introduce an analytic Gentzen-style calculus $\HE$ for $\E$, and use (a reformulation of) it to provide an alternative decidability proof for $\E$ and a complexity result. The calculus is also employed to generate formal explanations
for a well-known CTD paradox~\cite{Forr84} from the deontic logic literature. 

$\HE$ admits the elimination of the key rule of cut$-$which simulates Modus Ponens in Hilbert systems$-$and the consequent (relaxed version of the) subformula property; moreover its completeness proof
is independent from the logic's semantics. An "optimized" version $\HE+$ of $\HE$ is also given, that supports automated proof search and counterexample constructions.\footnote{
See~\cite{ddl:BFP19} for an alternative method for generating countermodels.} $\HE+$ is used to prove that the validity problem of $\E$ is co-NP and countermodels have polynomial size. 
%
%
%

We highlight two salient features of our approach. 
\vspace{-0.08cm}
\begin{itemize}
\item Since $\E$ is tightly connected  with the modal logic S5 ($\5$ is actually a sublogic of $\E$), our calculi are defined  using the hypersequent framework~\cite{Avr96}$-$a simple extension of Gentzen' sequent framework$-$needed to provide a cut-free calculus for $\5$~\cite{Min68,Avr96,Kurokawa13}, i.e. a calculus in which the cut rule is redundant.


\item Similarly to previous work on modal interpretation of conditionals, \emph{e.g.},~\cite{Nicola2,ben16},  we encode maximality by a unary modal operator.
Intuitively the fact that $x$ is among the best worlds that force a formula $A$ may be understood as saying that all the worlds accessible from $x$ via the betterness relation  
(or "above" $x$ according to the ranking) force not-$A$. This is 
 encoded as $\Bet \neg A$, where $\Bet$ is a K-type modal operator. The conditional obligation $\bigcirc(B/A)$ can be indirectly defined as $\Box (A\wedge \Bet\neg A  \rightarrow  B)$, where 
 $\Box$ obeys the laws of S5. 
Here "indirectly"  
indicates that the reduction schema is not explicitly introduced. $\Bet$ is not part of the language of $\E$,  but is used at the meta-level in the Gentzen-style system 
to define suitable rules for the conditional. 
\end{itemize}
We remark that although our calculus in some sense "translates" $\E$ into the bi-modal
logic $\5$+$K$, its complexity
turns out to be the same as for classical logic: co-NP; this contrasts
with the P-SPACE complexity of $\5$+$K$.

\vspace{-0.3cm}

	\section{System \E}
\label{preliminary}

\vspace{-0.2cm}

In this section we present the logic $\E$ both syntactically and semantically.

\begin{definition} The language $\mathcal{L}$ is defined by the following BNF:
$$A::= p\in \mbox{PropVar} \mid \neg A \mid A\ri A\mid \Box A\mid
\bigcirc (A/A)$$
$\Box A$ is read as ``$A$ is settled as true'', and
$\bigcirc (B/A)$ as ``$B$ is obligatory, given $A$''.
The Boolean connectives other than $\neg$ and $\ri$  are defined as usual.
\end{definition}

\begin{definition}
The axiomatization of $\E$ consists of  any Hilbert system for classical propositional logic,  the Modus Ponens rule $(MP)$:  If
$\vdash A \mbox{ and } \vdash A \rightarrow B\mbox{ then }\vdash B$, the rule $(Nec)$: $\mbox{If } \vdash A \mbox{ then } \vdash \Box A$ and the following axioms:  
\begin{flalign}
		& 
		\mbox{ \5 axioms for $\Box$   } 
		\tag{S5}\label{a5}
		\\&
		\bigcirc (B\IMPL C/A) \IMPL (\bigcirc (B/A) \IMPL \bigcirc
		(C/A))\tag{COK}\label{cok} \\ 
		& \bigcirc (A/A) \tag{Id}\label{id} \\ & \bigcirc (C/A\AND B) \IMPL
		\bigcirc (B\IMPL C/A)\tag{Sh}\label{sh} \\
		& \Box (A\IFF B) \IMPL (\bigcirc
		(C/A)\IFF \bigcirc (C/B) )\tag{Ext}\label{ext} \\ 
		& \bigcirc (B/A) \IMPL \Box\bigcirc
		(B/A) \tag{Abs}\label{abs} \\ & \Box A \IMPL \bigcirc (A/B)
		\tag{O-Nec}\label{nec} 
	\end{flalign}
The notions of derivation and theoremhood are as usual. 
\end{definition}
An intuitive reading of the axioms is as follows. A basic design choice of the logic $\E$ is that necessity 
is interpreted as in the modal logic S5. 
(\ref{cok}) is the conditional analogue of the familiar distribution
axiom K. (\ref{abs}) is the absoluteness axiom of~\cite{ddl:L73}, and
reflects the fact that the ranking is not world-relative.  (\ref{nec})
is the deontic counterpart of the  necessitation
rule. (\ref{ext}) permits the replacement of necessarily equivalent
sentences in the antecedent of deontic conditionals.  (\ref{id}) is
the deontic analogue of the identity principle. Named after Shoham~\cite[p.\,77]{ddl:S88} who seems to have been the
first to discuss it, (\ref{sh}) can be seen as expressing a "half" of deduction theorem or a "half" residuation property. The question of whether (\ref{id}) is a
reasonable law for deontic conditionals has been much debated (see~\cite{ddl:PS97} for a defense). 

The semantics of $\E$ can be defined in terms of \emph{preference models}. They  are possible-world models equipped with a comparative goodness relation $\succ$ on worlds so that  $x\succ y$ can be
	read as "world $x$ is \emph{better} than world $y$".  Conditional obligation is defined by considering  "best" worlds:  intuitively, $\bigcirc (B/A)$ holds in a model, if all the best worlds in which $A$ is true also make $B$ true.

\begin{definition}\label{def:0}  A preference model is a structure
$M=(W, \succ, V)$ ($W\not =\emptyset$)  whose members are called possible worlds, $\succ\subseteq W\times W$,  $V: W\rightarrow \mathcal{P}(PropVar)$. 
The following evaluation rules are used, for all $x\in W$: 
\begin{itemize}
	\item $M,x\vDash p$  iff  $p\in V(x) $ 
	\item $M,x\vDash  \neg A$  iff $ M,x\not\vDash A$ 
	\item $M,x\vDash  A\rightarrow B$  iff  if $ M,x\vDash A$  then $M, x\vDash B$
	\item $M,x\vDash \Box A \mbox{ iff }  \forall y \in W \  M,y\vDash A$ 
	\item $M,x\vDash \bigcirc (B/A) \mbox{ iff } \forall y\in \mathrm{best}(A) \ 	M, y\vDash B$
\end{itemize}
where $\mathrm{best}(A) = \{y\in W \mid M, y\models A \ \mbox{and there is no} \ z \succ y  \ \mbox{such that} \  M, z\models A\}$. 
	A formula $A$ is valid in a model $M$ 
	if for all worlds $x$ in $M$, $M,x\models A$. A formula $A$ is \emph{valid}  iff it is valid in every preference model. 
\end{definition}
Observe that we do not assume any  specific property of  $\succ$.

To the purpose of the calculi developed in the following, 
we introduce the modality $\mathit{\Bet}$, which will  allow us to represent the "Best" worlds: $M,x\vDash \mathit{\Bet} A$  iff  $\forall y\succ x \; M, y\vDash A$.
By this definition, we get $x\in \mathrm{best}(A)$ iff $M, x\models A$ and $M, x\models \Bet \neg A$.
However, the modality $\Bet$ is not part of  $\mathcal{L}$. %
As a  notational convention,  when no confusion  arise, we write $x\vDash A$ for  $M, x\vDash A$. 
The following result from~\cite{Parent15} is needed for subsequent developments:
\begin{theorem}\label{e:com}{\bf E} is sound and complete w.r.t.  the class of all preference models.
\end{theorem}
The completeness proof in~\cite{Parent15} uses another notion of maximality, call it $\mathrm{best'}$, where $y\in  \mathrm{best'}(A)$ iff $y\models A  \mbox{ and } \ y \succ z  \mbox{ for all } z\mbox{ s. t.} \  z\models A  \mbox{ and } \ z \succ y.$
Although  $\mathrm{best}$ and $ \mathrm{best'}$ are not equivalent, our result follows almost at once. Indeed, starting with a model $M=(W,\succ, V)$ in which obligations are evaluated using  $\mathrm{best'}$, one can derive an equivalent model $M'=(W,\succ', V)$ (with $W$ and $V$ the same) in which obligations are evaluated using $\mathrm{best}$.\footnote{Put $x\succ' y$ iff $x\succ y$ and $y\not\succ x$. We can easily verify that an arbitrarily chosen world satisfies exactly the same formulas in both models, viz. for all worlds $x$, $M,x\models A$ iff $M',x\models A$. (The sole purpose of this construction is to extend the result in~\cite{Parent15} to the current setting.)}

We end this section with two remarks.
The first one concerns reductions of conditional logics to modal logics.
In the literature various such reductions have been introduced; perhaps the best-known is the embedding of conditional logic into S4 put forth by Lamarre and Boutilier (see the discussion in \cite{ddl:M93} and the references therein). There are similarities with their approach, but also important differences. They define indeed an embedding of a conditional logic, different from $\E$, into S4.  In contrast, we do {\em not} embed $\E$ into any (bi)modal logic. $\E$  contains an $\5$ modality as a primitive notion, whose meaning is independent from the dyadic modality $\bigcirc (B/A)$.
%

The second remark concerns the suitability  of $\E$ to handle exceptions. 
Readers familiar with~\cite{ddl:PS97,ddl:horty14} may question this suitability.
We think that $\E$ does provide a minimal account of exceptions. However, we agree with~\cite{ddl:TT97} that a more adequate treatment of 
exceptions within a preference-based framework calls for the combined use of a normality relation and a betterness relation. 



	\section{A cut-free hypersequent calculus for $\E$} \label{an-e}

	We introduce the hypersequent calculus $\HE$ for the logic $\E$. $\HE$
	is defined in a modular way by adding to the calculus for the modal logic $\5$ suitable rules for the dyadic obligation, and the $\Bet$ operator.
	Introduced in~\cite{Min68} to define a cut-free calculus for $\5$, hypersequents consist of sequents working in parallel.
	
	
		\begin{definition}
			A \emph{hypersequent} is a multiset 
$\Gamma_1 \seq \Pi_1 \hh \dots \hh \Gamma_n \seq \Pi_n$
			where, for all $i = 1, \dots ,n,$ $\Gamma_i \seq \Pi_i$ is an
			ordinary sequent, called \emph{component}.
		\end{definition}
	
	The hypersequent calculus $\HE$ is presented in Def.~\ref{HE}.
	It consists of initial hypersequents (i.e., axioms), logical/modal/deontic and structural rules. The latter are divided into \emph{internal} and \emph{external rules}. 
		$\HE$ incorporates the sequent calculus for the modal logic S4 as a sub-calculus and adds an additional layer of information by considering a single sequent to live in the context of hypersequents. Hence all the axioms and rules of $\HE$ (but the external structural rules) are obtained by adding to each sequent a context $G$ (or $H$), representing a possibly empty hypersequent. For instance, the (hypersequent version of the) axioms are
	$\Gamma, p \Ri   \Delta, p \hh G $.
	%
		The external structural rules include
		ext. weakening (ew) and ext. contraction (ec) (see Fig.~\ref{h-rules}). These behave like weakening and contraction over whole hypersequent components. The hypersequent structure opens the possibility to
define new such rules that allow the "exchange of information" between different sequents.
It is this type of rules which increases the expressive power of hypersequent calculi compared to sequent calculi, allowing the definition of cut-free calculi for logics that seem to escape a cut-free sequent formulation (e.g., $\5$).
		An example of external structural rule is the $(s5)$ rule in~\cite{Kurokawa13} (reformulated as $(s5')$ in Fig.~\ref{h-rules} to account for the presence of $\bigcirc$), that allows the peculiar axiom of $\5$ to be derived as follows:
		
		\begin{figure}[!t]\small
	$$
		\begin{prooftree}
		\[\[\[\[
		\Box A \seq \Box \neg \Box A, \Box A
		\justifies
		 \Box A \seq \hh \seq \Box \neg \Box A, \Box A
		\using {\scriptscriptstyle \rm (s5')}
		\]
		\justifies
		\seq \neg \Box A \hh \seq \Box \neg \Box A, \Box A
		\using {\scriptscriptstyle \rm (\neg R)}
		\]
		\justifies
		\seq \Box \neg \Box A, \Box A \hh \seq \Box \neg \Box A, \Box A
		\using {\scriptscriptstyle \rm (\Box R)}
		\]
		\justifies
		\seq \Box \neg \Box A, \Box A
		\using {\scriptscriptstyle \rm (ec)}
		\]
		\justifies
		\seq \neg \Box A \to \Box \neg \Box A
		\using {\scriptscriptstyle \rm (\to  R) + (\neg L)}
		\end{prooftree}
		$$
\end{figure}

	\begin{figure}[!t]\small
			$$
				\infer[(ew)]{G \hh \Gamma \seq \Pi}{G}  \quad
				\infer[(ec)]{G \hh \Gamma \seq \Pi}{G \hh \Gamma \seq \Pi \hh \Gamma
					\seq \Pi}
			\quad 
				\infer[(s5')]{G \hh \Gamma  \seq  \hh \Gamma' \seq \Pi'}
				{G \hh \Gamma^{\square}, \Gamma^O, \Gamma' \seq \Pi'}
		$$
			\caption{External structural rules}
			\label{h-rules}
		\end{figure}		
			
	The rules in Fig.~\ref{h-rules} and~\ref{sh-rules} make use of the following notation:
	
	$
	\begin{array}{ll}
                \Sigma^{b\downarrow} = \{G : \Bet \; G\in \Sigma\} &
      \Sigma^{\Box} = \{\Box G  : \ \Box G \in \Sigma\} \\
                \Sigma^O = \{\bigcirc (C/D) : \bigcirc (C/D)\in \Sigma\} 
        \end{array}
   $
\begin{definition}
\label{HE}
The hypersequent calculus $\HE$ consists of the hypersequent version of Gentzen LK sequent calculus for propositional classical logic, the external structural rules in Fig.~\ref{h-rules} and the modal and deontic rules in Fig.~\ref{sh-rules}.
\end{definition}

	\begin{figure}[!t]\small
$$ \begin{array}{c c}
\irule{\Gamma^{\Box}, \Gamma^O, A, \Bet \ \neg A \Ri B\hh G}{\Gamma  \Ri \bigcirc (B / A), \Delta\hh G}{(\bigcirc R)}
&
\; \irule{\Gamma^{\Box},  \Gamma^O, \Gamma^{b\downarrow} \Ri  A \hh G }
{\Gamma \Ri   \Delta, \Bet \ A \hh G}{(\Bet)}\\\\

\irule{\Gamma^{\Box}, \Gamma^O  \Ri  A\hh G}
{\Gamma \Ri   \Delta, \Box A\hh G}{(\Box R)}
&
\irule{\Gamma, \Box A, A  \Ri \Delta\hh G}
{\Gamma, \Box A \Ri   \Delta\hh G}{(\Box L)}
\end{array}
$$
$$\irule{\Gamma, \bigcirc (B / A) \Ri   \Delta, A \hh G    \quad 
	\Gamma, \bigcirc (B / A) \Ri   \Delta, \Bet \  \neg A\hh G \quad 
	\Gamma, \bigcirc (B / A), B \Ri   \Delta\hh G}
{\Gamma, \bigcirc (B / A) \Ri  \Delta\hh G}{(\bigcirc L)} $$
\caption{Deontic and modal rules}
			\label{sh-rules}
\end{figure}

\noindent
A \emph{derivation} in $\HE$ is a tree obtained by applying the rules bottom up. A \emph{proof} $\cal D$ is a derivation whose  leafs are axioms. This distinction will be used in Sect.~\ref{alternative}.

The soundness of $\HE$ is proved with respect to preference models. Although we can interpret directly an hypersequent $H$ into the semantics, it is easier (and more readable) to interpret it as a formula  $I(H)$ of the extended language $\caL+\Bet$ and show the validity of this formula whenever $H$ is provable. 
		\begin{theorem}
		\label{Th:soundnessHE}
			If there is a proof in $\HE$ of $H: = \Gamma_1 \seq \Pi_1 \hh \dots \hh \Gamma_n \seq \Pi_n$, then $I(H):= \Box (\bigwedge \Gamma_1   \rightarrow \bigvee \Pi_1) \lor \ldots \lor \Box (\bigwedge \Gamma_n   \rightarrow \bigvee \Pi_n)$ is valid. 
		\end{theorem}
		\begin{proof}
By induction on the proof of $H$. We show  {\em $(\bigcirc R)$}, $(\Bet)$ and  {\em (s5')}.

\smallskip
\noindent
{\em $(\bigcirc R)$} Suppose  the premise is valid but not the conclusion. Thus for some model $M$ and  world  $x$,  $x\not\models \Box (\bigwedge\Gamma \rightarrow \bigvee \Delta \lor  \bigcirc (B / A) ) \lor \Box G$. Thus
$(1) \ x\not\models \Box (\bigwedge\Gamma \rightarrow \bigvee \Delta \lor  \bigcirc (B / A) )$ and $x\not\models \Box G$. 
	Since the premise is valid: $x\models \Box (\bigwedge\Gamma^{\Box} \land \bigwedge \Gamma^O \land A \land \Bet \ \neg A \ri B) \lor \Box G$
	so that	$(2) \ x\models \Box (\bigwedge\Gamma^{\Box} \land \bigwedge \Gamma^O \land A \land \Bet \ \neg A \ri B)$. From (1) there is $y$ s.t. $(3) \ y\models \bigwedge\Gamma$, $y\not\models \bigvee \Delta$ and  $y\not\models   \bigcirc (B / A)$;  for the latter there is some $z$ such that $ z\in best(A)$ and $z\not\models B$ [evaluation rule for $\bigcirc$]. So    $z\models A$  and  $z\models \Bet\neg A$ [def of $\Bet$]. 
	From (3), $y\models \bigwedge\Gamma^{\Box} \land \bigwedge \Gamma^O$, whence also for $z$, as  $\Gamma^{\Box}$ and  $\Gamma^O$ express  global assumptions, holding in all worlds in the model.
	Thus  $z\models \bigwedge\Gamma^{\Box} \land \bigwedge \Gamma^O \land A \land \Bet\neg A$. By (2)  $z\models B$, a contradiction.
    
\smallskip
\noindent
{\em $(\Bet)$}	
	Suppose that the premise is valid but not the conclusion. Thus for a model $M$ and  world $x$ 
	(ignoring the context $G$) 
	$x\not\models \Box (\bigwedge\Gamma \rightarrow \bigvee \Delta \lor \Bet A ) $, 
	but $(*) \ x\models \Box (\bigwedge\Gamma^{\Box} \land  \bigwedge\Gamma^O \land  \bigwedge\Gamma^{b\downarrow} \ri  A) $
	thus for some world $y$: (1) $y\models \bigwedge\Gamma$  (2) $y\not\models \Bet A$. Observe that 
	$(3) \ y \models \bigwedge\Gamma^{\Box} \land  \bigwedge\Gamma^O$
	and that (4) $y\models \Bet \ C$ for all $\Bet \ C\in\Gamma$.
	By (2) there is $z$ with $z\succ y$ s.t. $z\not \models A$.
	Hence $z\models \bigwedge\Gamma^{\Box} \land  \bigwedge\Gamma^O$. But by (4) we also get 
	$z\models \Gamma^{b\downarrow}$. Therefore  by (*) we get $z\models A$, a contradiction. 

\smallskip
\noindent
{\em (s5')}	
Suppose that the premise is valid but not the conclusion. Thus for some $M$ and   $x$,	$x\not\models \Box \neg \bigwedge\Gamma \lor  \Box (\bigwedge\Gamma' \rightarrow \bigvee \Pi')$, so  that $x\not\models \Box \neg \bigwedge\Gamma$ and  
		$x\not\models\Box (\bigwedge\Gamma' \rightarrow \bigvee \Pi')$. Therefore there are $y,z\in W$, such that $y\not\models \neg \bigwedge\Gamma$, meaning (1) $y\models  \bigwedge\Gamma$ and $z\not\models \bigwedge\Gamma' \rightarrow \bigvee \Pi'$, which entails (2) $z\models \bigwedge\Gamma'$ and (3)  $z\not\models \bigvee \Pi'$. By validity of the premise, $z\models \bigwedge\Gamma^{\Box} \land \bigwedge \Gamma^O \land \bigwedge \Gamma' \ri \bigvee \Pi'$, so that by (3), 
		(4) $z\not\models \bigwedge\Gamma^{\Box} \land \bigwedge \Gamma^O \land \bigwedge \Gamma'$. But by (1) , $z\models \bigwedge\Gamma^{\Box} \land \bigwedge \Gamma^O$ so that by (2) and (4) we have a contradiction. 
		\end{proof}
		\vspace{-0.2cm}
		\begin{theorem}[Completeness with cut]
		\label{th:completeness}
		Each  theorem of $\E$ has a proof in $\HE$ with the addition of 
		the cut rule: 
$$\infer[(cut)]{G \hh H \hh \Gamma, \Sigma \seq   \Delta, \Pi}
				{G \hh \Gamma, A  \seq \Delta \hspace{0.5cm} H \hh \Sigma \seq \Pi, A}$$
		\end{theorem}
		\vspace{-0.2cm}
		\begin{proof}
		As Modus Ponens corresponds
to the provability of $A, A \to B \seq B$ and two applications of cut, it suffices to show that $(Nec)$ and
all the axioms of $\E$ are provable in $\HE$. As an example, we show a
proof of $(\ref{cok})$:
 $$
 \small
\begin{prooftree}
 \[\[
 \[
 B \to C, B \seq C
\hspace{0.2cm}  \[ 
 A \seq A
\justifies
\Bet \; \neg A \seq \Bet \; \neg A
\using {\scriptscriptstyle \rm (\Bet)^\star}
\]
A \seq A
 \justifies
 B\IMPL C, \bigcirc (B/A), A, \Bet \; \neg A \seq C
 \using {\scriptscriptstyle \rm (\bigcirc  L)\ast}
\] 
\[ 
 A \seq A
\justifies
\Bet \neg A \seq \Bet \; \neg A
\using {\scriptscriptstyle \rm (\Bet)^\star}
\]
A \seq A
\justifies
\bigcirc (B\IMPL C/A), \bigcirc (B/A), A, \Bet \; \neg A \seq C
\using {\scriptscriptstyle \rm (\bigcirc  L)\ast}
                \]
                \justifies
\bigcirc (B\IMPL C/A), \bigcirc (B/A) \seq  \bigcirc (C/A)
\using {\scriptscriptstyle \rm (\bigcirc  R)}
                \]
                \justifies
\seq \bigcirc (B\IMPL C/A) \IMPL (\bigcirc (B/A) \IMPL \bigcirc (C/A))
\using {\scriptscriptstyle \rm (\to R) x 2}
                \end{prooftree}
$$

($\ast$ in the above proof stands for additional applications of internal weakening, and $(\Bet)^\star$ stands for $(\Bet) + (\neg L) + (\neg R)$)
		\end{proof}
		

		\subsection*{Cut-elimination}
		Theorem~\ref{th:completeness} heavily relies on the presence of
		the cut rule. In this section we give a constructive proof
		that cut can in fact be \emph{eliminated} from $\HE$ proofs.
		This result (cut elimination) implies (a relaxed form of) the {\it subformula property}:
		all formulas occurring in a cut-free $\HE$ proof are subformulas (possibly negated and under the scope of $\Bet$) of the formulas to be proved.

\vspace{-0.2cm}
\paragraph{Proof idea:}
To reduce the complexity of a cut on a formula of the form $\neg A$ or $A \to B$ we can exploit the rule invertibilities
(Lemma~\ref{invertibility}).
Some care is needed to deal with cut-formulas of the form $\Box A$, $\Bet \ A$ and $\bigcirc (B / A)$. There we cannot use the invertibility argument and cuts have to be shifted upward till the cut-formula is introduced. Notice however that the $(\Box R)$, 
$(\bigcirc R)$ and $(\Bet)$ rules do not allow to shift {\em every} cut upwards: only those involving sequents of a certain "good" shape. The proof hence proceeds by shifting uppermost cuts upwards in a specific order: first over the premise in which the cut formula appears on the right  (Lemma~\ref{reduction}) and then, when a rule introducing the cut formula is reached (and in this case the sequent has a "good" shape), shifting the cut upwards over the other premise (Lemma~\ref{leftinv}) till the left cut formula is introduced and the cut can be replaced by smaller cuts.
The hypersequent structure does not require major changes; as the $(s5')$ rule allows cuts with "good" shaped sequents to be shifted upwards, to handle $(ec)$ we consider the hypersequent version of the multicut: cutting one component (i.e. sequent) against possibly many components.

The {\em length} $\hgt{{\mathcal D}}$ of an $\HE$ proof ${\mathcal D}$
is (the maximal number of applications of inference rules) $+ 1$
occurring on any branch of~$d$.
The {\em complexity} $\cp{A}$ of a formula $A$ is  defined as: 
$\cp{A} = 0$ if $A$ is atomic, $\cp{\neg A} = \cp{A} +1$, 
$\cp{A \to B} = \cp{A}+ \cp{B} +1$, $\cp{\Bet \; A}= \cp{A} +1$,
$\cp{\Box A} = \cp{A} +1$,
and
$\cp{\bigcirc (A / B)} = \cp{A} + \cp{B} + 3$.
The {\em cut rank} $\rho({\mathcal D})$ of 
${\mathcal D}$ is
the maximal complexity + 1 of cut formulas in ${\mathcal D}$,  noting that
$\rho({\mathcal D}) = 0$ if ${\mathcal D}$ is cut-free. We use $A^n$ to indicate $n$ occurrences of $A$.

It is easy to see that the rules of the classical propositional connectives remain invertible, as stated in the lemma below.
\begin{lemma}
\label{invertibility}
Given an $\HE$ proof ${\cal D}$ of a hypersequent containing a compound formula $\neg A$ (resp. $A \to B$), we can find a proof ${\cal D}'$ of the same hypersequent ending in an introduction rule for $\neg A$ (resp. $A \to B$) and with
$\rho({\mathcal D}') \leq \rho({\mathcal D})$.
\end{lemma}
In $\HE$ any cut whose cut formula is immediately introduced in left and right premise can be replaced by smaller cuts. More formally,
\begin{lemma}
\label{reductivity}
Let $A$ be a compound formula and ${\mathcal D}_l$ and ${\mathcal D}_r$ be $\HE$ proofs such that $\rho({\mathcal D}_l) \leq \cp{A}$ and $\rho({\mathcal D}_r) \leq  \cp{A}$, and
\begin{enumerate}
\item   ${\mathcal D}_l$ is a proof of $\mg \h \Ga, A \seq \De$ ending in a rule introducing $A$ 
\item   ${\mathcal D}_r$ is a proof of $\mh \h  \Si \seq A, \Pi$ ending in  a rule introducing $A$ 
\end{enumerate}
We can find an $\HE$ proof of
$\mg \h  \mh \h  \Ga, \Si \seq \De, \Pi$  with $\rho({\mathcal D}) \leq \cp{A}$.
\end{lemma}
\begin{proof}
We show the only non-trivial case: $A = \Bet \ B$, where a cut

$$
\footnotesize
\begin{prooftree}
    \[
     H \h \Sigma^\Box, \Sigma^O, \Sigma^{b\downarrow}, B \seq C
	\justifies
		G \h \Sigma, \Bet \ B \seq \Bet \ C, \Pi
		\using {\scriptscriptstyle \rm (\Bet)}
		\]
    \[
     H \h \Gamma^\Box, \Gamma^O, \Gamma^{b\downarrow} \seq B
	\justifies
		H \h \Gamma \seq \Bet \ B, \Delta
		\using {\scriptscriptstyle \rm (\Bet)}
		\]
		\justifies
		G \h H \h \Gamma, \Sigma \seq \Bet \ C, \Delta, \Pi
		\using {\scriptscriptstyle \rm (cut)}
		\end{prooftree}
		$$
		is replaced by
		$$
		\footnotesize
\begin{prooftree}
    \[
    H \h \Sigma^\Box, \Sigma^O, \Sigma^{b\downarrow}, B \seq C
    \hspace{0.5cm}  G \h \Gamma^\Box, \Gamma^O, \Gamma^{b\downarrow} \seq B
	\justifies
		G \h H \h \Sigma^\Box, \Sigma^O, \Sigma^{b\downarrow} \Gamma^\Box, \Gamma^O, \Gamma^{b\downarrow} \seq C
		\using {\scriptscriptstyle \rm (cut)}
		\]
		\justifies
		G \h H \h \Gamma, \Sigma \seq \Bet \ C, \Delta, \Pi
		\using {\scriptscriptstyle \rm (\Bet)}
		\end{prooftree}
		$$

\end{proof}
\begin{lemma}
\label{leftinv}
Let ${\mathcal D}_l$ and ${\mathcal D}_r$ be $\HE$ proofs such that:
\begin{enumerate}
\item   ${\mathcal D}_l$ is a proof of $\mg \h \Ga_1, A^{\lam_1} \seq \De_1 
\h \dots \h \Ga_n, A^{\lam_n} \seq \De_n$;
\item   $A$ is a compound formula and ${\mathcal D}_r:= \mh \h  \Si \seq A, \Pi$ ends with a right logical rule introducing an indicated occurrence of $A$
\item   $\rho({\mathcal D}_l) \leq \cp{A}$ and $\rho({\mathcal D}_r) \leq  \cp{A}$;
\end{enumerate}
Then we can construct an $\HE$ proof ${\mathcal D}$ of
$\mg \h  \mh \h  \Ga_1, \Si^{\lam_1} \seq \De_1, \Pi^{\lam_1} \h \dots 
\h \Ga_n, \Si^{\lam_n} \seq \De_n, \Pi^{\lam_n}$ with $\rho({\mathcal D}) \leq \cp{A}$.
\end{lemma}
\begin{proof}
We distinguish cases according to the shape of $A$.  If $A$ is of the form $\neg B$ or $B \to C$, the claim follows by  Lemmas~\ref{invertibility}  and \ref{reductivity}. If $A$ is $\Box B$, $\bigcirc (B/C)$ or $\Bet \ B$ the proof 
proceeds by induction on $\hgt{{\mathcal D}_l}$. If ${\mathcal D}_l$ ends in an initial sequent, then we are done.
If ${\mathcal D}_l$ ends in a left rule introducing one of the indicated cut formulas, the claim follows by (i.h. and) Lemma~\ref{reductivity}.
Otherwise, let $(r)$ be the  last inference rule applied in ${\mathcal D}_l$. 
The claim follows by the i.h., an application of $(r)$ and/or weakening. Some care is needed to handle the cases in which $r$ is $(s5')$, $(\Box R)$, $(\bigcirc R)$ or $(\Bet)$ and $A$ is not in the hypersequent context $G$. Notice that when $A = \Box B$ (resp. $A = \bigcirc (B/C)$) the conclusion of ${\mathcal D}_r$ is $\Sigma \seq \Box B, \Pi$ (resp. $\Sigma \seq \bigcirc (B/C), \Delta$), but we can safely use the "good"-shaped sequent $\Sigma^\Box, \Sigma^O \seq \Box B$ (resp. $\Sigma^\Box, \Sigma^O \seq \bigcirc (B/C)$), that allows cuts to be shifted upwards over all $\HE$ rules, and we apply weakening afterwards. 
When $A = \Bet \ B$, notice that the cut formula does not appear in the premises of these rules.
For example let $(r) = (s5')$, $A = \Box B$, and ${\mathcal D}_l$ ends as follows
\[\footnotesize
      \begin{prooftree}
       \[
         \using
              d'_l
      \proofdotseparation=1ex
      \proofdotnumber=3
      \leadsto
\mg \h \Gamma^{\square}, \Box B, \Gamma^O, \Gamma' \seq \Pi' \h \dots \h \Omega, \Box B \seq \Delta
        \]
      \justifies
\mg \h  \Gamma, \Box B \seq \h \Gamma' \seq \Pi' \h \dots \h \Omega, \Box B \seq \Delta
        \using { \scriptstyle (s5')}
      \end{prooftree}
      \]
The claim follows by i.h. applied to the conclusion $\mg \h \Gamma^{\square}, \Box B, \Gamma^O, \Gamma' \seq \Pi' \h \dots \h \Omega, \Box B \seq \Delta$ of $d'_l$ (and $\Sigma^O, \Sigma^{\Box} \seq \square B$), followed by an application of $(s5')$ and weakening. The case $A = \bigcirc (B/C)$ is the same. 
The cases involving $(\Box R)$, 
$(\bigcirc R)$ and $(\Bet)$ are handled in a similar way.

\end{proof}

\begin{lemma}
\label{reduction}
Let ${\mathcal D}_l$ and ${\mathcal D}_r$ be $\HE$ proofs such that:
\begin{enumerate}
\item   ${\mathcal D}_l$ is a proof of $\mg \h  \Ga, A \seq \De$;
\item   ${\mathcal D}_r$ is a proof of $\mh \h  \Si_1 \seq  A^{\lam_1}, \Pi'_1  \h  
\dots \h \Si_n \seq  A^{\lam_n}, \Pi'_n$;
\item   $\rho({\mathcal D}_l) \leq \cp{A}$ and $\rho({\mathcal D}_r) \leq  \cp{A}$.
\end{enumerate}
Then  a proof ${\mathcal D}$ can be constructed in $\HE$  of
$\mg \h  \mh \h  \Si_1, \Ga^{\lam_1} \seq \Pi'_1, \De^{\lam_1}
\h  \dots \h  \Si_n, \Ga^{\lam_n} \seq \Pi'_n, \De^{\lam_n}$
with $\rho({\mathcal D}) \leq  \cp{A}$.
\end{lemma}
\begin{proof}
Let $(r)$ be the last inference rule applied in ${\mathcal D}_r$. If $(r)$ is an axiom, then the claim holds trivially. Otherwise, we proceed by induction on $\hgt{{\mathcal D}_r}$ , using Lemma~\ref{leftinv} when (one of) the indicated occurrence(s) of $A$ is principal. Assume $A$ is not principal. If $(r)$ acts only on $\mh$ or is a  rule other than $(s5')$, $(\Box R)$, $(\bigcirc R)$ and $(\Bet)$  the claim follows by the i.h. and an application of $(r)$. For the remaining rules notice that $A$ is not in the rule premise
(in case of $(s5')$ the "critical" component in the conclusion has empty right-hand side), hence the claim follows by applying (the i.h. to the other components, and) the respective rule followed by weakening.
\end{proof}

\begin{theorem}[Cut Elimination]
\label{th:cutelim}
Cut elimination holds for $\HE$.
\end{theorem}
\begin{proof}
Let ${\mathcal D}$ be an $\HE$ proof with $\rho({\mathcal D}) > 0$. We
proceed by a double induction on $\langle \rho({\mathcal D}), n\rho({\mathcal D}) \rangle$, where
$n\rho({\mathcal D})$ is the number of applications of cut in ${\mathcal D}$ with
cut rank $\rho({\mathcal D})$. Consider an uppermost application of $(cut)$ in
${\mathcal D}$ with cut rank $\rho({\mathcal D})$.
By applying  Lemma~\ref{reduction} to its premises
either $\rho({\mathcal D})$ or $n\rho({\mathcal D})$ decreases. 
\end{proof}

\begin{corollary}[Completeness]
Each  theorem of $\E$ has a proof in $\HE$.
\end{corollary}
\section{A proof search oriented calculus for $\E$}
\label{alternative}
The properties of the calculus $\HE$ include modularity, cut-elimination and a completeness proof which is independent from the semantics of $\E$. However $\HE$   supports neither automated proof search nor counterexample constructions.

Here we introduce the calculus $\HE^+$
having terminating proof search, 
thereby providing a decision procedure for $\E$, and in case of termination with failure
%
a countermodel of the starting formula can be extracted checking a \emph{single} failed derivation.
Similarly to the calculus for $\5$ in~\cite{DBLP:journals/igpl/KuznetsL16}, $\HE^+$ is obtained by making in $\HE$ all rules invertible, and all structural rules (including the external ones) admissible. Looking at the rules bottom up, this is achieved by copying the introduced formulas and the component containing it in the rule premises; the "simulation" of  
$(s5')$ is obtained by introducing additional left rules for $\Box$ and $\bigcirc (A/B)$ which add subformulas to different components of the hypersequent. 

Using $\HE^+$ we will show that the validity problem of $\E$ is co-NP. 

\begin{definition}
The $\HE^+$ calculus consists of: the initial hypersequents $\Gamma, p \Ri   \Delta, p \hh G $, together with the following rules: 

\begin{itemize}	
\item Rules for the propositional connectives that repeat the introduced formulas in the premises, for example

%

	$$\small \irule{\Gamma, A \ri B  \Ri \Delta, A \hh G \quad \Gamma, A \ri B, B  \Ri \Delta \hh G}
	{\Gamma,  A \ri B \Ri   \Delta \hh 
	G}{(\ri L)}  \qquad  
	\irule{\Gamma, A\Ri \Delta, A \ri B, B,  \hh G}
	{\Gamma \Ri   \Delta,  A \ri B \hh  G}{(\ri R)}$$
	\item  Rules for $O$

		$$\small\irule{\Gamma \Ri \bigcirc (B / A),\Delta \hh A, \Bet \ \neg A \Ri B \hh G}{\Gamma  \Ri \bigcirc (B / A), \Delta\hh G}{(\bigcirc R+)}$$
		
		$$\small\irule{ \Gamma, \bigcirc (B / A) \Ri   \Delta, A \hh G    \quad 
			\Gamma, \bigcirc (B / A) \Ri   \Delta, \Bet \  \neg A\hh G \quad 
			\Gamma, \bigcirc (B / A), B \Ri   \Delta\hh G}
		{\Gamma, \bigcirc (B / A) \Ri  \Delta\hh G}{(\bigcirc L+)}$$
		
		{$$\scriptsize\irule{ \Gamma, \bigcirc (B / A) \Ri \Delta \hh  \Sigma \Ri \Pi, A \hh G    \quad 
			\Gamma, \bigcirc (B / A) \Ri   \Delta  \hh  \Sigma \Ri \Pi, \Bet \  \neg A\hh G \quad 
			\Gamma, \bigcirc (B / A) \Ri   \Delta \hh \Sigma, B \Ri \Pi\hh G}
		{\Gamma, \bigcirc (B / A) \Ri  \Delta \hh \Sigma \Ri \Pi \hh G}{(\bigcirc L2)}$$}

	\item  Rule for $\Bet$	

		$$\small \irule{\Gamma\Ri \Delta, \Bet \ A \hh \Gamma^{b\downarrow} \Ri  A \hh G }
		{\Gamma \Ri   \Delta, \Bet \ A \hh G}{(\Bet+)}$$
		
		\item  Rules for $\Box$	
		
		$$\small \irule{\Gamma\Ri \Delta, \Box A \hh  \Ri  A\hh G}
		{\Gamma \Ri   \Delta, \Box A\hh G}{(\Box R+)}
		\quad
		 \irule{\Gamma, \Box A, A  \Ri \Delta\hh G}
		{\Gamma, \Box A \Ri   \Delta\hh G}{(\Box L+)}
		\quad  
		 \irule{\Gamma, \Box A \Ri \Delta \hh \Sigma, A \Ri \Pi \hh G}
		{\Gamma, \Box A \Ri   \Delta \hh \Sigma \Ri \Pi \hh G}{(\Box L2)}$$
		
	\end{itemize} 
	\end{definition}
	The notion of proof and derivation is as for $\HE$.
The following lemma collects standard structural properties of $\HE^+$.

\begin{lemma}

\label{Lem:Admiss}
(i) All rules of  $\HE^+$ are height-preserving invertible. (ii)  rules applications  permute over each other (with the usual exceptions).
(iii) Internal and external weakening and contraction are admissible in  $\HE^+$.
\end{lemma}
\begin{proof}
(i) Follows by the fact that the premises already contain the conclusion. (ii) and (iii) are standard (and hence omitted). 
\end{proof}
As a consequence of this lemma the order of application of the rules is 
irrelevant.
\vspace{-0.2cm}	\begin{theorem}
If there is a proof of $H$ in $\HE^+$ then $I(H)$ is valid
\end{theorem}

\begin{proof}
    We first show that the rules of $\HE^+$ can be simulated in $\HE$. This holds
    for all the $\HE^+$ rules but $(\bigcirc L 2)$ and $(\Box L 2)$ 
    by simply applying weakening, internal and external contraction. For $(\Box L2)$ we have

\begin{prooftree}\scriptstyle
        \[\[\[ \[
G \h \Gamma, \Box A \seq \Delta \h \Sigma, A \seq \Pi
                \justifies
G \h \Gamma, \Box A \seq \Delta \h \Sigma, \Box A \seq \Pi
\using
    {\scriptscriptstyle \rm (\Box L )}
                \]
                \justifies
G \h \Gamma, \Box A \seq \Delta \h \Box A \seq  \h \Sigma \seq \Pi
\using
    {\scriptscriptstyle \rm (S5')}
                \]
                \justifies
G \h \Gamma, \Box A \seq \Delta \h \Gamma, \Box A \seq \Delta \h \Sigma \seq \Pi
    \using
{\scriptscriptstyle \rm (w)}
            \]
                \justifies
G \h \Gamma, \Box A \seq \Delta \h \Sigma \seq \Pi
\using {\scriptscriptstyle \rm (ec)}\]
		\end{prooftree}

    The argument for $(\bigcirc L2)$ is analogous.
    The claim follows by Theorem~\ref{Th:soundnessHE}. 
\end{proof}
We have adopted a "kleene'd" formulation of the calculus to make easier countermodel construction and termination of proof-search.
They are both based on the notion of \emph{saturation} that we define next. 
Given a hypersequent $H$, we write $\Gamma \Ri \Delta\in H$ to indicate that $\Gamma \Ri \Delta$ is a compontent of $H$.

\begin{definition}[Saturation]
	A hypersequent $H$ is \emph{saturated} if it is not  an axiom and satisfies the following conditions associated to each rule application 
	\begin{description}
		\item[$(\ri L)_S$] if  $\Gamma, A \ri B  \Ri \Delta \in H$ then either $A\in \Delta$ or $B\in \Gamma$
		\item[$(\ri R)_S$] if  $\Gamma\Ri \Delta, A \ri B   \in H$ then  $A\in \Gamma$ and $B\in \Delta$
		\item[$(\neg L)_S$] if  $\Gamma, \neg A \Ri \Delta \in H$ then  $A\in \Delta$
		\item[$(\neg R)_S$] if  $\Gamma\Ri \Delta, \neg A  \in H$ then  $A\in \Gamma$ 
		\item[$(\bigcirc L+)_S$] if  $\Gamma, \bigcirc (B / A) \Ri  \Delta   \in H$ then either $A\in \Delta$ or $\Bet \neg A\in \Delta$ or $B\in \Gamma$
		\item[$(\bigcirc L2)_S$] if  $\Gamma, \bigcirc (B / A)  \Ri  \Delta   \in H$ and $\Sigma\Ri\Pi\in H$ then either $A\in \Pi$ or $\Bet \neg A\in \Pi$ or $B\in \Sigma$
		\item[$(\bigcirc R+)_S$] if  $\Gamma  \Ri \bigcirc (B / A), \Delta   \in H$ then there is $\Sigma\Ri\Pi\in H$ such that $A\in \Sigma$,  $\Bet \neg A\in \Sigma$, and  $B\in \Pi$
		\item[$(\Bet+)_S$] if $\Gamma \Ri   \Delta, \Bet A \in H$ then there is $\Sigma\Ri\Pi\in H$ such that $\Gamma^{b\downarrow} \subseteq \Sigma$ and $A\in \Pi$
		\item[$(\Box R+)_S$] if $\Gamma \Ri   \Delta, \Box A \in H$ then there is $\Sigma\Ri\Pi\in H$ such that $A\in \Pi$
		\item[$(\Box L+)_S$] if $\Gamma, \Box A \Ri   \Delta  \in H$ then $A\in\Gamma$
		\item[$(\Box L2)_S$] if $\Gamma, \Box A \Ri   \Delta  \in H$ and $\Sigma\Ri\Pi\in H$ then   $A\in\Sigma$
	\end{description}
\end{definition}	

The key to obtain termination is to avoid the application of a rule to hypersequents which in a sense already contain the premise of that rule. 

\begin{definition}[Redundant application]
	A backward application of a rule $(R)$ to an hypersequent $H$ is \emph{redundant} if $H$ satisfies the saturation condition $(R)_S$ associated to that application of $(R)$.
\end{definition}

We call a derivation/proof \emph{irredundant} if (i) no rule is applied to an axiom, and (ii) it does not contain any redundant application of rule. It is easy to see that by the admissibility  of internal weakening and external contraction (Lemma~\ref{Lem:Admiss}) redundant applications of the rules can be safely removed.

%

\begin{lemma}
	 Every hypersequent provable in $\HE^+$ has an irredundant proof.
\end{lemma}
\begin{proof}
By induction on the height of a uppermost redundant application. To illustrate the argument consider a redundant application of the $(\Bet+)$ rule
	$$\irule{\Gamma\Ri \Delta, \Bet \ A \hh \Gamma^{b\downarrow} \Ri  A \hh \Gamma^{b\downarrow},\Sigma' \seq \Pi', A \hh G}
{\Gamma \Ri   \Delta, \Bet \ A \hh \Gamma^{b\downarrow},\Sigma' \seq \Pi', A \hh G}{(\Bet+)}$$

this is transformed as follows: \\

\begin{prooftree}
\[	\Gamma\Ri \Delta, \Bet \ A \hh \Gamma^{b\downarrow} \Ri  A \hh \Gamma^{b\downarrow},\Sigma' \seq \Pi', A \hh G
\justifies
\Gamma\Ri \Delta, \Bet \ A \hh\Gamma^{b\downarrow},\Sigma' \seq \Pi', A \hh \Gamma^{b\downarrow},\Sigma' \seq \Pi', A \hh G
	\using (Wk)
\]
	\justifies
	\Gamma\Ri \Delta, \Bet \ A  \hh \Gamma^{b\downarrow},\Sigma' \seq \Pi', A \hh G
	\using (ec)
\end{prooftree}

\vspace{0.5cm}
\end{proof}

The above property justifies the restriction to irredundant proofs from a syntactical point of view, although this justification is not really needed for completeness (Theorem \ref{semcom} below).

We now use the calculus $\HE^+$ to give a decision procedure for the logic $\E$; the key issue here is to restrict proof-search to irredundant derivations. 

We denote  by $|A|$ the size of a formula $A$ considered as a string of symbols.

\begin{theorem}\label{finite}
\label{bound}
	Every $\HE^+$ derivation of a formula $A$ of $\E$ is finite and it
	is either a proof or it contains a saturated hypersequent. 
\end{theorem}
\begin{proof}
Let ${\mathcal D}$ be any derivation built from $\seq A$ by backwards application of the rules.  We first prove that all hypersequents in ${\mathcal D}$ are finite and provide an upper bound on their size. To this purpose let  $|A|= n$ and consider
	$SUB^+(A)   = \{B \mid B \ \mbox{is a subformula of $A$ }\}  \cup \{\Bet \neg C \mid \bigcirc (D / C)  \ \mbox{occurs in } A, \mbox{for some} \ C\}$.
Clearly the cardinality of 	$SUB^+(A)$ is $O(n)$ and so it is  the size of each formula in it.

Let $H :=  	\Gamma_1 \seq \Delta_1 \hh \dots \hh \Gamma_k \seq \Delta_k$ be any hypersequent occurring in  ${\mathcal D}$.
The size of each component is bounded by $O(n^2)$: it contains $O(n)$ formulas each one of size $O(n)$. To estimate the size of $H$, we estimate the number of its components (i.e. $k$).
Observe that the rules which "create" new components are $(\Box R+)$, $(\bigcirc R+)$, and $(\Bet+)$.
Consider first $(\Box R+)$: by the irredundancy restriction this rule is applied \emph{exactly once} to each formula, say $\Box C$, occurring in the consequent of a component and creates only \emph{one} new component, no matter if $\Box C$ appears in the consequent of many components. To illustrate the situation, consider, e.g.,
	$$\irule{\ldots\Gamma_i \Ri \Delta_i,  \Box C\hh \seq C  \hh \ldots \hh \Gamma_j \Ri \Delta_j, \Box C\hh \ldots \hh \Gamma_k \Ri   \Delta_k}
{\ldots\Gamma_i \Ri   \Delta_i, \Box C\hh \ldots \hh \Gamma_j \Ri   \Delta_j, \Box C \hh \ldots \hh \Gamma_k \Ri   \Delta_k}{}$$
the irredundancy restriction 
ensures that if $(\Box R+)$ is applied to  $\Gamma_i \Ri   \Delta_i$,  it cannot be applied to the component $\Gamma_j \Ri   \Delta_j, \Box C$. 
This means that the number of components created by $(\Box R+)$ is bounded by $\Box$-ed subformulas of $A$, whence  it is $O(n)$. 
The situation for $(\bigcirc R+)$ is similar.

For the rule $(\Bet+)$, first observe the following fact:\\
{\it Given any derivation  ${\mathcal D}$ having at its root a formula of $\E$ (that is a hypersequent $\seq A$) at most one $\Bet$ formula  can occur in the antecedent $\Gamma_i$ of any component of any hypersequent  in ${\mathcal D}$, that is  $\Gamma_i^{b\downarrow}$ contains at most one formula. }

By this fact the rule $(\Bet+)$  may be applied when  $\Gamma_i^{b\downarrow}$ contains a formula and when  $\Gamma_i^{b\downarrow}=\emptyset$, in both cases the applications are not duplicated,  
for instance in the former case, we may have: 
	$$\small \irule{\ldots\Gamma_i, \Bet \ \neg E \Ri   \Bet \ \neg F \hh \neg E \seq \neg F  \hh \ldots \hh \Gamma_j, \Bet \ \neg E \Ri   \Bet \ \neg F\hh \ldots }
{\ldots\Gamma_i, \Bet \ \neg E \Ri  \Bet \ \neg F \hh \ldots \hh \Gamma_j, \Bet \ \neg E \Ri  \Bet \ \neg F\hh \ldots }{}$$
Thus there is at most one application of the $(\Bet+)$ rule for any pair of $\Bet$ formulas (case  $\Gamma_i^{b\downarrow}\not=\emptyset$) plus possibly an application for any $\Bet$ formula (case $\Gamma_i^{b\downarrow}=\emptyset$).
Since $\Bet$ formulas come from the decomposition of $O$-subformulas and there are $O(n)$  of them, the number of components created by the $(\Bet+)$ rule is $O(n^2+n) = O(n^2)$.
We can conclude that the number of components of any hypersequent in ${\mathcal D}$ is  $O(n^2)$, whence  the size of each hypersequent is $O(n^4)$.

We get also an upper bound on proof branches:  since any backward application of a rule is irredundant,  it must add some formula/component. Therefore the length of each proof branch is also bounded by $O(n^4)$ and the derivation is finite. Finally each leaf must be an axiom or a saturated hypersequent otherwise a rule would have been applied to it.
\end{proof}

The next theorem  shows the completeness of  $\HE^+$.

\begin{theorem}\label{semcom}
Every valid formula $A$ of $\E$ has a proof in $\HE^+$.
\end{theorem}
\begin{proof}
	We prove the contrapositive: if $A$ is not provable in $\HE^+$ then there is a model in which $A$ is not valid. 
	Suppose that $A$ is not provable in $\HE^+$, by the previous theorem any derivation of $\seq A$ as root contains at least one branch ending with a  saturated hypersequent. Fix a derivation and let \\
	$H =  	\Gamma_1 \seq \Delta_1 \hh \dots \hh \Gamma_n \seq \Delta_n$
	be the intended saturated hypersequent. We build a countermodel of $A$ based on $H$.
	First we enumerate the components of $H$, calling $H'$ the corresponding structure: 
		$$H' = 1: \Gamma_1 \seq \Delta_1 \hh 2: \Gamma_2 \seq \Delta_2  \dots \hh n: \Gamma_n \seq \Delta_n$$
	We then define a model $M = (W, \succ, V)$ by stipulating:
	\vspace{-0.13cm}
	\begin{quote}
		$W = \{1,\ldots, n\}, \quad$
		$V(i) = \{P \mid P\in \Gamma_i\}$ with   $i: \Gamma_i \seq \Delta_i\in H' $\\
		$j \succ i$ where $i: \Gamma_i \seq \Delta_i\in H' $, $j: \Gamma_j \seq \Delta_j\in H' $, we have $\Gamma_i ^{b\downarrow} \subseteq \Gamma_j$  and there is a formula $\Bet \ C\in \Delta_i$ such that $C\in\Delta_j$.
	\end{quote}
	\vspace{-0.13cm}
	Notice that in the definition of the preference relation it may be $i=j$.
	We now prove the fundamental claim (truth lemma); to this purpose we do not need to consider formulas with $\Bet$, thus for $B\in \caL$:
	\begin{itemize}
		\item[(a)] for any $i\in W$, if $B\in\Gamma_i$ then $M,i\models B$
		\item[(b)] for any $i\in W$, if $B\in\Delta_i$ then $M,i\not\models B$
	\end{itemize}
	Both claims (a) and (b) are proved by structural induction on $B$.
	\begin{itemize}
		\item Let $B$ be an atom $P$, then (a) holds by definition of $V(i)$. Concerning (b), let $P\in\Delta_i$, since $H$ is saturated,  $P\not\in\Gamma_i$, otherwise $H$ would be an axiom; thus $P\not\in V(i)$ whence $M,i\not\models P$.
		\item the propositional case use saturation conditions and induction hypothesis. 
		\item Let $B = \bigcirc (D/C)$. (a) suppose $ \bigcirc (D / C)\in \Gamma_i$. We have to show that for every $j\in W$ the following holds: (case 1)  $M,j\not\models C$,  or (case 2) there is $k\in W$ with  $k\succ j$ such that $M,k\models C$, or (case 3)  $M,j\models D$.
		By saturation conditions $(\bigcirc L+)_S$ or  $(\bigcirc L2)_S$ according to $i=j$ or $i\not=j$, we have that either  $C\in \Delta_j$ or $\Bet \neg C\in \Delta_j$ or $D\in \Gamma_j$, in the first case by i.h. we get $M,j\not\models C$ (case 1), in the third case, by i.h.  we get $M,j\models D$ (case 3). Thus we are left with the case $\Bet \neg C\in \Delta_j$. By saturation condition $(\Bet)_S$, there is $k: \Gamma_k\seq\Delta_k\in H'$ such that $\Gamma_j^{b\downarrow} \subseteq \Gamma_k$
		and $\neg C\in \Delta_k$. Observe that by construction it holds $k\succ j$. Moreover,  by saturation condition $(\neg R)_S$, $C\in \Gamma_k$, whence by inductive hypothesis $M,k\models C$.\\
		(b) Suppose  $ \bigcirc (D / C)\in \Delta_i$. We have to show that there is  $j\in W$ such that: $M,j\models C$; for all $k\in W$ with $k \succ j$  $M,k\not\models C$; and $M,j\not\models D$. 
		By  $(O R)_S$ there is $j:\Gamma_j\seq\Delta_j\in H'$  such that $C\in \Gamma_j$,  $\Bet \neg C\in \Gamma_j$, and  $D\in \Delta_j$; by i.h. we get $M,j\models C$ and $M,j\not\models D$. We have still to prove that for all $k\in W$ with $k\succ j$  $M,k\not\models C$. To this aim, suppose $k \succ  j$, by construction we have that there is   $j: \Gamma_k \seq \Delta_k\in H' $ such that 
		$\Gamma_j ^{b\downarrow} \subseteq \Gamma_k$  and for some   formula $\Bet \ E\in \Delta_j$ it holds $E\in\Delta_k$. Since $\Bet \neg C\in \Gamma_j$,   $\neg C\in \Gamma_j ^{b\downarrow} \subseteq \Gamma_k$, whence by  $(\neg L)_S$ $C\in \Delta_k$; by i.h. we conclude $M,k\not\models C$ and we are done. 
		\item $B = \Box C$. (a) suppose $ \Box C\in \Gamma_i$. We have to show that for every $j\in W$, $M,j\models C$. Let $j\in W$ this means that  $j: \Gamma_k \seq \Delta_k\in H' $ (it might be $j=i$), by saturation condition $(\Box L+)_S$ or $(\Box L2)_S$, according to $i=j$ or $i\not= j$ we have $C\in \Gamma_j$, whence by i.h.  $M,j\models C$.\\
		(b) Suppose $ \Box C\in \Delta_i$. By saturation condition $(\Box R+)_S$  there is $j: \Gamma_j\Ri\Delta_j\in H'$ such that $C\in \Delta_j$, thus by i.h. $M,j\not\models C$.
	\end{itemize}
Being $\seq A$ the root of the derivation, for some $i: \Gamma_i \seq \Delta_i\in H'$, we have $A\in\Delta_i$, and by claim~(b) $M,i\not\models A$, showing that $A$ is not valid in $M$.
\end{proof}

This allows us to obtain a complexity bound for validity in  $\E$. 

\begin{theorem}
 Validity of formula of $\E$ can be decided in Co-NP time.
\end{theorem}
\begin{proof}
Given $A$,  to decide whether $A$ is valid, we consider a non-deterministic algorithm which takes as input  $\seq A$ and guesses a saturated hypersequent $H$: if it 
finds it, the algorithm answers "non-valid", otherwise, it  answers "valid". As shown in the proof of the Theorem~\ref{bound}, the size of the candidate saturated hypersequent $H$ is  polynomially bounded by the size of $A$ (= $O(|A|^4)$), moreover checking whether  $H$ is saturated can also be done in polynomial time in the size of $A$ (namely $O(|A|^8)$). More concretely, the algorithm can try to build  
$H$ by applying the rules backwards in an arbirary but fixed order, applying the first applicable (i.e. non-redundant) rule and then choosing non-deterministically  one of its premises if there are  more than one. The number of steps is polynomially bounded  by $O(|A|^4)$ and checking whether a rule is applicable  to a given  hypersequent  is linear in  the size  of the hypersequent. Thus the whole non-deterministic computation is polynomial  in the size of the input formula. 
\end{proof}

By the previous results $\E$ turns out to have 
the polysize model property.

\begin{corollary}
If a formula $A$ of $\E$ is satisfiable (that is $\neg A$ is not valid), then it has a model of polynomial size in the length of $A$.
\end{corollary}

\hide{
\begin{example}
\marginpar{Xavier????}
	We build a countermodel of the formula 
	$$(\Box (P \ri Q) \land O(R \mid  Q)) \ri O(R \mid P)$$
Let start a derivation with root $\seq F$, where $F$ is the above formula. Applying  the rules backwards we generate (among others) the  hypersequent \\
	$\Gamma_1 \seq\Delta_1 \hh \Gamma_2 \seq\Delta_2 \hh \Gamma_3 \seq\Delta_3$
	where : 
	 	\begin{eqnarray*}
	 	\Gamma_1 \seq\Delta_1 & = & \Box (P \ri Q),  O(R \mid  Q), P \ri Q,\seq O(R \mid P),  P, Q\\
	 	\Gamma_2 \seq\Delta_2 & = & P \ri Q,  Q, P, \Bet \neg P\seq R, \Bet\neg Q\\
	 	\Gamma_3 \seq\Delta_3 & = &  P \ri Q,  \neg P, Q, R \seq \neg Q, P,  \Bet\neg Q
	 \end{eqnarray*}
	  This hypersequent is saturated. 
	  Following the construction of 
	  Theorem~\ref{semcom}, we enumerate the components (respectively) by 1,2,3 and get the model $M = (W, >, V)$ where $W = \{1,2,3\}$, the preference relation is $3 > 2$ and $3>3$, and $V(1)= \emptyset, V(2) = \{P,Q\}, V(3) = \{Q,R\}$. It is easy to see that $i\models P\ri Q$, for $i=1,2,3$, $O(R \mid  Q)$ holds in the model, but $O(R \mid P)$ does not as $2\models P$, for all $j>2$ (i.e. $j=3$), $j\not\models P$, and $2\not\models R$.
\end{example}
}		

We end the section  with an example of explanation, obtained by countermodel construction, of a well-known CTD paradox. 

\vspace{0.2cm}

\noindent {\bf "Gentle Murder"}~\cite{Forr84}. 
Consider the following norms and fact:  
(i) You ought not kill
(ii) If you kill, you ought to kill gently
(iii) Killing gently is killing
(iv) You kill. 
In many deontic logics, these sentences are inconsistent and in particular (ii)-(iv) allow to derive the obligation to kill, contradicting (i)--hence the "paradox". 
We formally show that this does not happen in the logic $\E$.
To this purpose let the above sentences be encoded by: $\bigcirc (\neg k / \top), \bigcirc (g / k), \Box(g\ri k), k$ with the obvious meaning of propositional atoms. 
We first verify that the above formulas are consistent, thus we begin a derivation with root hypersequent
$$ \bigcirc(\neg k / \top), \bigcirc (g / k), \Box(g\ri k), k \seq \bot$$
One of the saturated hypersequents we find by applying the rules backwards is
$$ \bigcirc(\neg k / \top), \bigcirc (g / k), \Box(g\ri k), k, g\ri k, g \seq \bot, \Bet \neg\top \hh g\ri k, \neg k \seq \neg \top, g$$
Following the construction of Theorem~\ref{semcom}, we enumerate the components (respectively) by 1,2 and get the model $M = (W, \succ, V)$ where $W = \{1,2\}$, the preference relation is $2 \succ 1$  and $V(1) = \{k,g\},V(2)= \emptyset$. It is easy to see that $i\models g\ri k$, for $i=1,2$, both $\bigcirc (\neg k / \top), \bigcirc (g / k)$ are valid in the model and $1\models k$. Notice in particular that 1 is the "best" world where "kill" holds and in that world also "killing gently" holds. 

We can also verify that the sentences (ii)-(iv) do not derive the obligation to kill. Notice that this claim in $\E$ is not entailed by what we have just proved. 
To this purpose we initialise the derivation by 
$ \bigcirc (g / k), \Box(g\ri k), k \seq \bigcirc (k / \top)$	 
and we get (among others) the following saturated hypersequent:
$$\bigcirc (g / k), \Box(g\ri k), k, g\ri k, g \seq \bigcirc (k / \top) \hh \top, \Bet\neg \top, g\ri k \seq k, g$$
We get the model $M = (W, \succ, V)$, where $W$ and $V$ are as before (1 and 2 are now constructed using the new hypersequent), but $\succ$ is empty meaning that all worlds are best. 
Now 2 is a  "best" world in an absolute sense (i.e., for $\top$) and $k$ does not hold there.  By the evaluation rule 
(cf. Def.~\ref{def:0}), $\bigcirc (k / \top)$ fails both in 1 and 2.
Hence
killing is not best overall, and you are not obliged to kill.

\hide{		
		\begin{example}
Another example is Chisholm's paradox~\cite{chisholm} that consists of the following four sentences:
1. You ought to go to the assistance of your neighbours
2. If you go to the assistance of your neighbours, you ought to  let them know that
you are coming
3. If you do not go to the assistance of your neighbours, you ought to not tell them
that you are coming
4. You are not going to the assistance of your neighbour
		\end{example}
}

\subsection*{Acknowledgements}
Work funded by the projects FWF M-3240-N and  WWTF MA16-028. 
We thank the anonymous reviewers for their valuable comments.



\end{document}